\documentclass{ifacconf}
\usepackage{stfloats}
\usepackage{graphicx}      
\usepackage{natbib}        
\usepackage{pdfsync}
\usepackage{epsfig} 
\usepackage{cuted}
\usepackage{lipsum}
\usepackage{mathtools}
\usepackage{amsmath,amssymb,amsfonts}
\usepackage{mathrsfs}
\usepackage{enumerate}
\usepackage{color}
\usepackage{subfigure}
\usepackage{comment}
\usepackage{algorithm}
\usepackage{flushend}
\usepackage{algpseudocode}

\newtheorem{mythm}{Theorem}

\newtheorem{mylem}{Lemma}

\newtheorem{myexm}{Example}
\newtheorem{remark}{Remark}
\begin{document}
\begin{frontmatter}

\title{Data-Driven Min-Max MPC with Integral Quadratic Constraints\thanksref{footnoteinfo}}

\thanks[footnoteinfo]{F. Allg\"{o}wer is thankful that his work was funded by Deutsche Forschungsgemeinschaft (DFG, German Research
Foundation) under Germany’s Excellence Strategy - EXC 2075 - 390740016 and under grant 468094890.
The authors thank the International Max Planck Research School
for Intelligent Systems (IMPRS-IS) for supporting Yifan Xie.}

\author[First]{Yifan Xie}
\author[First]{Julian Berberich}
\author[First]{Frank Allg\"{o}wer}

\address[First]{University of Stuttgart, Institute for Systems Theory and Automatic Control, 70550 Stuttgart, Germany (e-mail: \{yifan.xie, julian.berberich, frank.allgower\}@ist.uni-stuttgart.de).}

\begin{abstract}    
Data-driven control of nonlinear systems with rigorous guarantees is a challenging control problem.
Integral quadratic constraints (IQCs) provide a powerful framework for modeling nonlinearities.
This paper presents a data-driven min-max model predictive control (MPC) synthesis method for unknown systems subject to (nonlinear) uncertainties using the IQC framework.
The unknown system matrices are characterized by a set-membership representation using the input-state data and the knowledge of the IQCs.
We derive two semidefinite programs (SDPs) that minimize an upper bound on the worst-case cost over all possible system dynamics and uncertainties.
By iteratively solving these SDPs, the proposed state-feedback control law is obtained.
We further prove that the resulting closed-loop system is exponentially stable and satisfies the input and state constraints.
A numerical example demonstrates the validity of
the proposed method.
\end{abstract}

\end{frontmatter}

\section{Introduction}\label{sec1}
Classical controller design techniques often require a precise system model. 
However, first-principles models can be difficult to obtain and usually require expert knowledge.
Data-driven control is becoming more and more popular recently, as it relies on available data rather than explicit system models for controller synthesis.
Within the data-driven control literature, several frameworks have been proposed for controller design, e.g., data informativity framework \citep{van2023informativity} and data-driven model predictive control (MPC) based on Willems' Fundamental Lemma \citep{markovsky2021review, markovsky2023datapower}.
The data informativity framework assumes that the noise satisfies energy bounds or instantaneous bounds and proposes a quadratic matrix inequality to characterize a set of linear-time-invariant (LTI) systems that are consistent with the data.
The controller design then aims to synthesize a controller that stabilizes all systems that are consistent with the data.
Although this framework is initially designed for LTI systems, it has been extended to linear parameter varying systems \citep{verhoek2024decoupling,xie2024dataLPV}, polynomial systems \citep{martin2023guarantees}, and bilinear systems \citep{xie2025bilinear}.
However, data-driven controller synthesis for systems with nonlinear uncertainties remains largely unexplored within this framework.

The integral quadratic constraints (IQC) framework provides a powerful tool for modeling structured uncertainties and
nonlinearities \citep{veenman2016robust}.
The IQC framework is shown to represent a wide range of uncertainties, including time-invariant and time-varying parametric uncertainties, norm-bounded nonlinearities, slope-restricted nonlinearities, time-delay uncertainties, and others.
It enables an efficient analysis and controller synthesis of robust stability and performance for uncertain dynamical systems based on linear matrix inequality (LMI) optimization. 
Robust controller synthesis methods based on IQCs for discrete-time systems have been studied by \cite{veenman2014iqc,hu2017robustness, fry2017iqc}.
These works all assume that the system model is known and that the uncertainties satisfy prescribed IQCs.
Data-driven system analysis problems such as controllability, dissipativity, and stabilizability are discussed by \cite{markovsky2021review}.
Data-driven controller synthesis is studied by \cite{GUPTA2026112608}, which develops a data-driven controller synthesis approach for generalized multi-input multi-output systems under IQC-modeled perturbations using only frequency-domain measurement data.
\cite{xie2024dataLPV, berberich2023combining} propose data-driven controller design methods but restrict attention to multiplicative nonlinearities and impose pointwise constraints on the uncertainties.
A data-driven controller is designed for nonlinear systems using input–state data by \cite{luppi2022nonlinear}; however, it requires the nonlinearity to be measurable.

MPC is an advanced control technique that can handle nonlinear systems and input/state constraints.
The LMI-based min-max MPC scheme designs control inputs to minimize the worst-case cost over uncertainties to robustly stabilize the system \citep{kothare1996robust}.
Recently, data-driven MPC approaches have been studied,
which directly use the measured input–output data to predict the future behaviors based on Willems' Fundamental Lemma \citep{coulson2019data,berberich2021guarantees,markovsky2023datapower}.
Within the data informativity framework, data-driven min–max MPC schemes have been proposed to account for performance objectives and constraints using LMI-based methods: for LTI systems by \cite{nguyen2023lmirobust,xie2024minmax,xie2024minmaxrobust}, for LPV systems by \cite{xie2024dataLPV}, and for bilinear systems by \cite{xie2025bilinear}.
A min–max robust formulation of data-driven MPC is proposed by \cite{wang2025min} to optimize worst-case performance over the uncertainty set by analyzing non-unique solutions to the behavioral representation.

In this paper, we consider an unknown discrete-time uncertain system described by the interconnection of a nominal LTI system and an uncertainty satisfying a known IQC.
We propose a data-driven min-max MPC scheme using only input-state data.
An algorithm is developed that iteratively solves two SDPs and obtains an upper bound on the data-driven min-max MPC cost and the corresponding state-feedback gain.
A receding-horizon algorithm is used to improve the performance by repeatedly applying the iterative algorithm online.
The proposed scheme guarantees exponential stability of the closed-loop system and satisfaction of both input and state constraints. 
A numerical example further demonstrates the effectiveness of the proposed method. 
In contrast to existing data-driven control approaches based on IQCs, such as \cite{GUPTA2026112608}, our method uses time-domain data and accounts for input and state constraints.

The remainder of this paper is organized as follows. Section~\ref{sec2} introduces the system dynamics and states the problem setup.
In Section~\ref{sec3}, we present the data-driven characterization using the input-state data and the IQCs.
We propose the data-driven min-max MPC scheme with IQCs and prove the closed-loop guarantees.
Section~\ref{sec4} demonstrates the implementation of the proposed method through a numerical example.
Finally, Section~\ref{sec5} concludes the paper and outlines future research directions.

\emph{Notations: }
For matrices $A, B$ with compatible dimensions, we abbreviate $A^\top BA$ to $[\star]^\top B A$.
If a matrix $P$ is positive (semi-)definite, we write $P\succ 0$($P\succeq 0$).
For a vector $x$ and a matrix $P\succ 0$, we write $\|x\|_P=\sqrt{x^\top P x}$.

\section{Problem Setup}\label{sec2}




This paper considers a discrete-time uncertain system described by the interconnection of a nominal linear time-invariant (LTI) system and a (possibly nonlinear) uncertainty $\Delta$.
The state-space model of the system is described as follows:
\begin{equation}\label{system}
\begin{aligned}
    x_{t+1}&=A_s x_t+B_s u_t+B_\omega \omega_t,\\
    z_t&=Cx_t+Du_t,\\
    \omega_t&=\Delta(z_t),
\end{aligned}
\end{equation}
where $x_t\in\mathbb{R}^{n_x}$ is the state,  $u_t\in\mathbb{R}^{n_u}$ is the input, $\omega_t\in\mathbb{R}$ is the uncertainty, and $\Delta: \mathbb{R}\rightarrow \mathbb{R}$ is a potentially nonlinear operator.
We assume that the matrices $C\in\mathbb{R}^{1\times n_x}$, $D\in\mathbb{R}^{1\times n_u}$ and $B_\omega\in\mathbb{R}^{n_x}$ are known, while $\Delta$ is assumed to be unknown.
The knowledge of $C, D, B_\omega$ means that we have some information about which parts of dynamics are affected by $\Delta$, although the explicit expression of the operator 
$\Delta$ is not available. Moreover, the nominal LTI system is unknown, i.e., $A_s\in\mathbb{R}^{n_x\times n_x}$,  $B_s\in\mathbb{R}^{n_x\times n_u}$ are not known.

\begin{remark}
In the existing literature, \cite{xie2024dataLPV, berberich2023combining} consider multiplicative nonlinearity $\omega_t=\Delta_t z_t$ and impose a pointwise quadratic constraint on $\Delta_t$.
\cite{luppi2022nonlinear} assumes that the nonlinearity is measurable.
In this paper, we generalize to
more general nonlinearities and do not require the nonlinearity to be known precisely.
\end{remark}

We assume that the operator $\Delta$ satisfies the following time-domain integral quadratic constraint (IQC).

\begin{assum}\label{assumption1}
For all signals $\{z_t, \omega_t\}_{t=0}^{T-1}$ with $\omega_t=\Delta(z_t)$ and all $T\in\mathbb{N}$, it holds that 
\begin{equation}\label{assumption1_equation}
\sum_{t=0}^{T-1} \begin{bmatrix}\omega_t &z_t\end{bmatrix}\Pi \begin{bmatrix}\omega_t &z_t\end{bmatrix}^\top\geq 0,
\end{equation}
where $\Pi\in\boldsymbol\Pi$ for a known set $\boldsymbol\Pi$. We denote the set of operators $\Delta$ satisfying \eqref{assumption1_equation} by $\boldsymbol\Delta$.
\end{assum}

\begin{myexm}
Consider the thermal regulation dynamics of a building modeled as $x_{t+1}=A_sx_t+B_su_t+\omega_t$, where $x_t$ is the temperature deviation from a desired setpoint, $u_t$ is the heat flux, and $\omega_t=\Delta(x_t)$ is a (potentially nonlinear) uncertainty.
The uncertainty $\Delta$ is not known but is constrained by a sector, i.e.,  $(\Delta(x)-\alpha x)(\beta x -\Delta(x))\geq 0$. Then \eqref{assumption1_equation} holds with 
\begin{equation}\label{piuncertainty}
\boldsymbol{\Pi}=\left\{\begin{bmatrix}
   1  &-\alpha\\
   -1 &\beta
\end{bmatrix}^\top\!\!\begin{bmatrix}
    0 &G_{12}\\
    G_{12}^\top &0\end{bmatrix}\!\!\begin{bmatrix}
   1  &-\alpha\\
   -1 &\beta
\end{bmatrix}:G_{12}\geq 0\right\}.
\end{equation}
\end{myexm}




We assume that a sequence of data measurements $(U, X, Z)$ of length $T_f$ from the system \eqref{system} is available.
The available data are denoted as follows
\begin{equation}\label{data}
\begin{aligned}
U\!\!=\!\!\begin{bmatrix}u_0^d \!&\ldots \!&u_{T_f-1}^d\end{bmatrix},
X\!\!=\!\!\begin{bmatrix}x_0^d \!&\ldots \!&x_{T_f}^d\end{bmatrix},
Z\!\!=\!\!\begin{bmatrix}z_0^d \!&\ldots \!&z_{T_f-1}^d\end{bmatrix}.
\end{aligned}
\end{equation}

Our goal is to design a time-varying state-feedback controller $u_t=F_tx_t$ with $F_t\in\mathbb{R}^{n_x\times n_u}$ that stabilizes the origin for the closed-loop system using the data \eqref{data}.
In order to evaluate the closed-loop performance, we define the stage cost function as $\ell(x, u)=\|x\|_Q^2+\|u\|_R^2$, where $Q, R\succeq 0$.
The considered input and state constraints are formulated as $\|x\|_{S_x}^2\leq 1$ and $\|u\|_{S_u}^2\leq 1$, where $S_u\succ 0, S_x\succeq 0$.

\section{Data-driven min-max MPC with IQCs}\label{sec3}

In Section~\ref{sec3.1}, we present a data-driven system characterization of the unknown system matrices $(A_s, B_s)$.
Then, we propose the data-driven min-max MPC scheme with IQCs in Section~\ref{sec3.2}.
Closed-loop guarantees of the proposed scheme are shown in Section~\ref{sec3.3}.

\subsection{Data-driven system characterization}\label{sec3.1}

We define the set of $(A, B)$ that are consistent with the data \eqref{data} and the knowledge on the operator $\Delta$ as
\begin{equation}
    \Sigma\!=\!\!\left\{\!(A, B):\!\!
    \begin{gathered}
        x_{i+1}^d=Ax_i^d+Bu_i^d+B_\omega\Delta(z_i^d), 
        \\ \text{holds for some }\Delta\in\boldsymbol{\Delta}, \forall i\in\mathbb{I}_{[0, T_f-1]} 
    \end{gathered}\right\}.
\end{equation}

In the following lemma, we propose a data-driven characterization of the set $\Sigma$ using the data \eqref{data} and Assumption~\ref{assumption1}.

\begin{mylem}\label{lemma1}
Suppose Assumption~\ref{assumption1} holds and $B_\omega\neq 0$. Then, the set $\Sigma$ is equal to
\begin{equation}\label{consistent_set}
\left\{
(A, B):
\begin{gathered}
    \begin{bmatrix}
        I
        &A
        &B
    \end{bmatrix}
    S(\tilde{\Pi})
    \begin{bmatrix}
        I
        &A
        &B
    \end{bmatrix}^\top\succeq 0, \forall \tilde{\Pi}\in\boldsymbol{\Pi}
\end{gathered}
\!\right\}
\end{equation}
with \begin{equation}
    S(\tilde{\Pi})=\sum_{i=0}^{T_f-1}\!
    \begin{bmatrix}
        x_{i+1}^d &B_\omega z_i^d\\
        -x_i^d &0\\
        -u_i^d &0
    \end{bmatrix}\tilde{\Pi}
    \begin{bmatrix}
        x_{i+1}^d &B_\omega z_i^d\\
        -x_i^d &0\\
        -u_i^d &0    \end{bmatrix}^\top.
\end{equation} 
\end{mylem}
\begin{pf}
Since the IQC \eqref{assumption1_equation} holds for all $T\in\mathbb{N}$, we have
    $\sum_{i=0}^{T_f-1}\begin{bmatrix}\Delta(z_i^d) &z_i^d\end{bmatrix}\tilde{\Pi} \begin{bmatrix}\Delta(z_i^d) &z_i^d\end{bmatrix}^\top\geq 0$
for $\tilde{\Pi}\in\boldsymbol{\Pi}$.
Left- and right multiplying the above inequality with $B_\omega$ and $B_\omega^\top$, respectively, we obtain
\begin{equation}\label{pf_2}
    \sum_{i=0}^{T_f-1} \begin{bmatrix}B_\omega \Delta(z_i^d) &B_\omega z_i^d\end{bmatrix}\tilde{\Pi} \begin{bmatrix}B_\omega \Delta(z_i^d) &B_\omega z_i^d\end{bmatrix}^\top\succeq 0.
\end{equation}
Based on the system dynamics \eqref{system}, $(A, B)\in\Sigma$ if and only if \eqref{pf_2} holds with $B_\omega\Delta (z_i^d)=x_{i+1}^d-Ax_i^d-Bu_i^d$.
Replacing $B_\omega\Delta (z_i^d)$ by $x_{i+1}^d-Ax_i^d-Bu_i^d$ in \eqref{pf_2}, the set \eqref{consistent_set} includes all possible $(A, B)$ that are consistent with the data \eqref{data}.

To prove the reverse inclusion, let $(A, B)$ be in the set \eqref{consistent_set}.
Then, we have \eqref{pf_2} holds by defining $B_\omega\Delta (z_i^d)=x_{i+1}^d-Ax_i^d-Bu_i^d$.
Since $B_\omega\neq 0$, $\Delta\in\boldsymbol{\Delta}$.
Since $x_{i+1}^d=Ax_i^d+Bu_i^d+B_\omega\Delta (z_i^d)$ holds for some $\Delta\in\boldsymbol{\Delta}$ for all $i\in\mathbb{I}_{[0, T_f-1]}$, $(A, B)\in \Sigma$.
$\hfill\qed$
\end{pf}

\subsection{Data-driven min-max MPC}\label{sec3.2}

Given the current state $x_t\in\mathbb{R}^{n_x}$, the data-driven min-max MPC problem is formulated as follows:
\begin{subequations}\label{mpc}
\begin{align}
J_\infty^*(x_t)&:=\min_{\bar{u}(t)}\max_{(A, B)\in\Sigma, \Delta\in\boldsymbol{\Delta}}\sum_{k=0}^{\infty}\ell(\bar{u}_k(t), \bar{x}_k(t))\label{mpc:obj}\\
\text{s.t.}\quad &\bar{x}_{k+1}(t)=A\bar{x}_k(t)+B\bar{u}_k(t)+B_\omega \Delta(\bar{z}_k(t)),\label{mpc:con1}\\
&\bar{z}_k(t)=C\bar{x}_{k}(t)+D\bar{u}_k(t),\label{mpc:con2}\\
&\bar{x}_0(t)=x_t,\label{mpc:con3}\\
&\|\bar{u}_k(t)\|_{S_u}\leq 1, \forall k\in\mathbb{N},\label{mpc:con4}\\
&\|\bar{x}_k(t)\|_{S_x}\leq 1, \forall (A, B)\in\Sigma, \Delta\in\boldsymbol{\Delta}, k\in\mathbb{N},\label{mpc:con5}
\end{align}
\end{subequations}
where $\bar{u}_k(t), \bar{x}_k(t)$ denote the predicted input and state at time $t+k$ given the current measured state $x_t$, $\bar{u}$ is a sequence of predicted inputs with $\bar{u}_{k}(t)=F\bar{x}_{k}(t), \forall k\in\mathbb{N}$.
The objective of this optimization problem is to find a state-feedback control law that minimizes the worst-case infinite-horizon cost over all consistent system matrices in $\Sigma$ and all $\Delta\in\boldsymbol{\Delta}$.
The predicted input and state satisfy the constraints for any $(A, B)\in\Sigma, \Delta\in\boldsymbol{\Delta}$ as in \eqref{mpc:con4}-\eqref{mpc:con5}.
The data-driven min-max MPC problem is intractable because of the min-max formulation and the infinite-horizon cost. 
We will later formulate two computationally tractable problems and iteratively solve the problems to obtain a state-feedback control law that stabilizes the system with $(A, B)\in\Sigma$ and $\Delta\in\boldsymbol{\Delta}$.

Given the state $x_t\in\mathbb{R}^{n_x}$ and the data \eqref{data}, the optimization problem is first formulated as
\begin{subequations}\label{sdp1}
    \begin{align}
    &\min_{\gamma>0, P\in\mathbb{R}^{n_x\times n_x}, F\in\mathbb{R}^{n_u\times n_x}, \Pi, \tilde{\Pi}\in\boldsymbol{\Pi}}\gamma\\
         \text{s.t. }&x_t^\top P x_t\leq \gamma,\label{sdp1.1}\\
    &\begin{bmatrix}\star\end{bmatrix}^\top\!\!\!
    \begin{bmatrix}
        -P\!+\!Q\!+\!F^\top\!R F \!\!\!& \!\!\!& \!\!\!&\\
           \!\!\!&P \!\!\!& \!\!\!&\\
           \!\!\!&  \!\!\!&\Pi \!\!\!&\\
           \!\!\!&  \!\!\!&  \!\!\!&M(\tilde{\Pi})
    \end{bmatrix}\!\!\!\!
    \begin{bmatrix}
        I &0 &0\\
        0 &B_\omega &I\\
        0 &I &0\\
        C\!+\!DF &0 &0\\
        0 &0 &I\\
        \begin{bmatrix}I\\F\end{bmatrix} &0 &0
    \end{bmatrix}\prec 0,\label{sdp1.2}\\
    &P\succeq \gamma F^\top S_u F, P\succeq \gamma S_x, \label{sdp1.3}
    \end{align}
\end{subequations}
where $M(\tilde{\Pi})=\begin{bmatrix}I &0\\0 &-I\end{bmatrix}S(\tilde{\Pi})^{-1}\begin{bmatrix}-I &0\\0 &I\end{bmatrix}$.
In Theorem~1, we show that the optimal cost of \eqref{mpc} is bounded by $\gamma$.

\begin{mythm}\upshape
Suppose that Assumption~\ref{assumption1} holds, $B_\omega\neq 0$ and there exist $\gamma>0, P\in\mathbb{R}^{n_x\times n_x}, F\in\mathbb{R}^{n_u\times n_x}, \Pi, \tilde{\Pi}\in \boldsymbol{\Pi}$ such that \eqref{sdp1.1} and \eqref{sdp1.2} hold.
Then, $\gamma$ is an upper bound on the optimal cost of the data-driven min-max MPC problem \eqref{mpc}.
\end{mythm}
\begin{pf}
Since we focus on designing a state-feedback controller $\bar{u}_{k}(t)=F\bar{x}_{k}(t), \forall k\in\mathbb{N}$, the predicted closed-loop system dynamics with any $(A, B)\in\Sigma$ can be written as
\begin{equation}\label{closed_system}
\begin{aligned}
\bar{x}_{k+1}(t)&\!=\!\bar{\omega}_{AB, k}(t)+B_\omega \bar{\omega}_k(t),\\
\bar{z}_{AB, k}(t)&=\begin{bmatrix}
    I& F^\top
\end{bmatrix}^\top\bar{x}_k(t),
\bar{\omega}_{AB, k}(t)\!=\!\begin{bmatrix}A &B\end{bmatrix}\bar{z}_{AB, k}(t),\\
\bar{z}_k(t)&\!=\!(C+DF)\bar{x}_k(t),
\bar{\omega}_k(t)\!=\!\Delta(\bar{z}_k(t)),
\end{aligned}
\end{equation}
where the system parametric uncertainty $(A, B)\in\Sigma$ enters the dynamics through the channel from $\bar{z}_{AB, k}(t)\in\mathbb{R}^{n_x+n_u}$ to $\bar{\omega}_{AB, k}(t)\in\mathbb{R}^{n_x}$.

Multiplying the inequality \eqref{sdp1.2} from left and right by 
$\begin{bmatrix}\bar{x}_k(t)^\top &\bar{\omega}_k(t)^\top &\bar{\omega}_{AB, k}(t)^\top\end{bmatrix}^\top$ and its transpose, respectively,  and using \eqref{closed_system}, the following inequality holds
\begin{equation}\label{proof1}
    \begin{bmatrix}\star\end{bmatrix}^\top \begin{bmatrix}
        -P\!+\!Q\!+\!F^\top\!R F \!\!\!& \!\!\!& \!\!\!&\\
           \!\!\!&P \!\!\!& \!\!\!&\\
           \!\!\!&  \!\!\!&\Pi \!\!\!&\\
           \!\!\!&  \!\!\!&  \!\!\!&M(\tilde{\Pi})
    \end{bmatrix}\!\!\!\!
    \begin{bmatrix}
        \bar{x}_k(t)\\
        \bar{x}_{k+1}(t)\\
        \bar{\omega}_k(t)\\
        \bar{z}_k(t)\\
        \bar{\omega}_{AB, k}(t)\\
        \bar{z}_{AB, k}(t)
    \end{bmatrix}\leq 0,
\end{equation}
which further leads to
\begin{equation}\label{proof2}
\begin{aligned}
    -\|\bar{x}_k(t)\|_P^2+\ell(\bar{u}_k(t), \bar{x}_k(t))+\|\bar{x}_{k+1}(t)\|_P^2\\+\begin{bmatrix}
    \bar{\omega}_k(t)\\
    \bar{z}_k(t)
    \end{bmatrix}^\top \!\!\!\!\Pi\!\!\begin{bmatrix}
    \bar{\omega}_k(t)\\
    \bar{z}_k(t)
    \end{bmatrix}
    \!\!+\!\!\begin{bmatrix}
        \bar{\omega}_{AB, k}(t)\\
        \bar{z}_{AB, k}(t)
    \end{bmatrix}^\top \!\!\!\!\!M(\tilde{\Pi})\!\!
    \begin{bmatrix}
        \bar{\omega}_{AB, k}(t)\\
        \bar{z}_{AB, k}(t)
    \end{bmatrix}\!\!\leq\! 0.
\end{aligned}
\end{equation}
Summing \eqref{proof2} up from $k=0$ to $k=T$ and letting $T\rightarrow \infty$, 
\begin{equation}\label{proof21}
\begin{aligned}
    \|\bar{x}_\infty(t)\|_P^2-\|\bar{x}_0(t)\|_P^2+\sum_{k=0}^\infty \ell (\bar{u}_k(t), \bar{x}_k(t))+\\
    \sum_{k=0}^\infty \!\!\begin{bmatrix}
    \bar{\omega}_k(t)\\
    \bar{z}_k(t)
    \end{bmatrix}^\top \!\!\!\!\Pi\!\!\begin{bmatrix}
    \bar{\omega}_k(t)\\
    \bar{z}_k(t)
    \end{bmatrix}
    \!\!+\!\!\sum_{k=0}^\infty\!\! \begin{bmatrix}
        \bar{\omega}_{AB, k}(t)\\
        \bar{z}_{AB, k}(t)
    \end{bmatrix}^\top \!\!\!\!M(\tilde{\Pi})\!\!
    \begin{bmatrix}
        \bar{\omega}_{AB, k}(t)\\
        \bar{z}_{AB, k}(t)
    \end{bmatrix}\leq 0.
\end{aligned}
\end{equation}
Since \eqref{assumption1_equation} holds for any $T\in\mathbb{N}$, $\omega=\Delta(z)$ and we have $ \bar{z}_k(t), \bar{\omega}_k(t) \in\mathbb{R}$, 
the following inequality holds 
\begin{equation}\label{proof3}
    \sum_{k=0}^\infty\begin{bmatrix}
    \bar{\omega}_k(t)\\
    \bar{z}_k(t)
    \end{bmatrix}^\top \!\!\!\!\Pi\!\!\begin{bmatrix}
    \bar{\omega}_k(t)\\
    \bar{z}_k(t)
    \end{bmatrix}\geq 0.
\end{equation}
Since $\begin{bmatrix}
        I
        &A
        &B
    \end{bmatrix}
    S(\tilde{\Pi})
    \begin{bmatrix}
        I
        &A
        &B
    \end{bmatrix}^\top\succeq 0$ holds for any $(A, B)\in\Sigma$,
using the Dualization lemma in \cite[Lemma 4.9]{scherer2000linear}, we have
\begin{equation}\label{proof4}
    \begin{bmatrix}
        \begin{bmatrix}A &B\end{bmatrix}\\
        I
    \end{bmatrix}^\top M(\tilde{\Pi}) \begin{bmatrix}
        \begin{bmatrix}A &B\end{bmatrix}\\
        I
    \end{bmatrix}\succeq 0.
\end{equation}
Multiplying \eqref{proof4} from left and right by $\bar{z}_{AB, k}(t)^\top$ and its transpose, respectively, using $\bar{\omega}_{AB, k}(t)=\begin{bmatrix}A &B\end{bmatrix}\bar{z}_{AB, k}(t)$, we have
\begin{equation}\label{thm1proof1}
    \begin{bmatrix}
        \bar{\omega}_{AB, k}(t)\\
        \bar{z}_{AB, k}(t)
    \end{bmatrix}^\top
    M(\tilde{\Pi})
    \begin{bmatrix}
        \bar{\omega}_{AB, k}(t)\\
        \bar{z}_{AB, k}(t)
    \end{bmatrix}\succeq 0.
\end{equation}
Summing \eqref{thm1proof1} up from $k=0$ to $T$ and letting $T\rightarrow \infty$, we have
\begin{equation}\label{proof5}
\sum_{k=0}^\infty
\begin{bmatrix}
        \bar{\omega}_{AB, k}(t)\\
        \bar{z}_{AB, k}(t)
    \end{bmatrix}^\top \!\!\!\!M(\tilde{\Pi})\!\!
    \begin{bmatrix}
        \bar{\omega}_{AB, k}(t)\\
        \bar{z}_{AB, k}(t)
    \end{bmatrix}\geq 0.
\end{equation}
Using \eqref{proof21} together with \eqref{proof3} and \eqref{proof5}, we have 
\begin{equation}\label{proof51}
    \|\bar{x}_\infty(t)\|_P^2+\sum_{k=0}^\infty \ell (\bar{u}_k(t), \bar{x}_k(t))\leq \|\bar{x}_0(t)\|_P^2
\end{equation}
for all $(A, B)\in\Sigma, \Delta\in\boldsymbol{\Delta}$.
Since $\|\bar{x}_\infty(t)\|_P^2\geq 0$, $\bar{x}_0(t)=x_t$ and using \eqref{sdp1.1} and \eqref{proof51}, we obtain
$\max_{(A, B)\in\Sigma, \Delta\in\boldsymbol{\Delta}}\sum_{k=0}^\infty \ell (\bar{u}_k(t), \bar{x}_k(t))\leq \|x_t\|_P^2\leq \gamma$.
Thus, $\gamma$ is an upper bound on the optimal cost of the data-driven min-max MPC problem \eqref{mpc}.
$\hfill\qed$
\end{pf}

The problem \eqref{sdp1} is computationally intractable as its constraints are not linear in the variables $F$ and $\tilde{\Pi}$.
Given the state $x_t\in\mathbb{R}^{n_x}$ and the data \eqref{data}, we formulate the following equivalent optimization problem
\begin{subequations}\label{sdp2}
    \begin{align}
        &\max_{\lambda>0, H\in\mathbb{R}^{n_x\times n_x}, L\in\mathbb{R}^{n_u\times n_x}, \Pi, \tilde{\Pi}\in\boldsymbol{\Pi}} \lambda\\
\text{s.t.} &H-\lambda x_t x_t^\top\succeq 0,\label{sdp2.1}\\
& \begin{bmatrix}
            \begin{bmatrix}\star\end{bmatrix}^\top\!\!\!
    \begin{bmatrix}
        H \!\!\!\!& \!\!\!\!&\\
        \!\!\!\!&\Pi^{-1} \!\!\!\!&\\
        \!\!\!\!& \!\!\!\!&M(\tilde{\Pi})^{-1}
    \end{bmatrix}\!\!\!\!
    \begin{bmatrix}
        I &0 &0\\
      -B_\omega^\top &0 &0\\
        0 &I &0\\
        -I &0 &0\\
        0 &0 &I
    \end{bmatrix} &\begin{bmatrix}
        0\\
        \!CH\!\!+\!\!DL\!\\
        \begin{bmatrix}
            H\\
            L
        \end{bmatrix}
    \end{bmatrix} &0\\
    \begin{bmatrix}
        0 &(CH\!\!+\!\!DL)^\top &\begin{bmatrix}H\\L\end{bmatrix}^\top
    \end{bmatrix}
    &H &\Phi^\top\\
    0 &\Phi &I
        \end{bmatrix}\succ 0,\label{sdp2.2}\\
    &\begin{bmatrix}H &L^\top\\L &\lambda S_u^{-1}\end{bmatrix}\succeq 0, 
    \begin{bmatrix}
        H &H\\ H &\lambda S_x^{-1}
    \end{bmatrix}\succeq 0,\label{sdp2.3}
    \end{align}
\end{subequations}
where $\Phi=\begin{bmatrix}M_RL\\M_Q H\end{bmatrix}$ with $M_R^\top M_R=R, M_Q^\top M_Q=Q$ and $M(\tilde{\Pi})$ is defined the same as \eqref{sdp1}.
Problem \eqref{sdp2} is not linear in $\Pi$, but linear in other variables.
We will later propose an algorithm that iteratively solves these two problems by fixing some variables to find the optimal solution of problem \eqref{sdp1}.

In the following theorem, we show that the problem \eqref{sdp1} is equivalent to the problem \eqref{sdp2}.

\begin{mythm}
Problem \eqref{sdp1} is equivalent to the problem \eqref{sdp2}.
\end{mythm}
\begin{pf}
Let $H=P^{-1}$ and $\lambda=\frac{1}{\gamma}$.
Using the Schur complement twice, \eqref{sdp1.1} is equivalent to \eqref{sdp2.1}.
The objective of \eqref{sdp1} is equivalent to that of \eqref{sdp2} given that $\lambda=\frac{1}{\gamma}$.

Using the Dualization lemma in \cite{scherer2000linear}, \eqref{sdp1.2} is equivalent to
\begin{equation}\label{proof6}
\begin{aligned}
    \begin{bmatrix}\star\end{bmatrix}^\top
    \begin{bmatrix}
        (Q\!+\!F^\top R F\!-\!P)^{-1}& & &\\
        &P^{-1} & &\\
        & &\Pi^{-1} &\\
        & & &M(\tilde{\Pi})^{-1}
    \end{bmatrix}\!\!\\
    \begin{bmatrix}
        0 \!\!&-(C\!\!+\!\!DF)^\top \!\!&-\!\begin{bmatrix}I\\ F\end{bmatrix}^\top\\
        I \!\!&0 \!\!&0\\
      -B_\omega^\top \!\!&0 \!\!&0\\
        0 \!\!&I \!\!&0\\
        -I \!\!&0 \!\!&0\\
        0 \!\!&0 \!\!&I
    \end{bmatrix}\!\!\!\succ\! 0.
\end{aligned}
\end{equation}
The inequality \eqref{proof6} is equivalent to
\begin{equation}\label{proof7}
\begin{aligned}
    \begin{bmatrix}
        0\\
        C+DF\\
        \begin{bmatrix}
            I\\
            F
        \end{bmatrix}
    \end{bmatrix}
    (Q+F^\top RF-P)^{-1}
    \begin{bmatrix}
        0\\
        C+DF\\
        \begin{bmatrix}
            I\\
            F
        \end{bmatrix}
    \end{bmatrix}^\top+\\
    \begin{bmatrix}\star\end{bmatrix}^\top
    \begin{bmatrix}
        P^{-1} & &\\
        &\Pi^{-1} &\\
        & &M(\tilde{\Pi})^{-1}
    \end{bmatrix}
    \begin{bmatrix}
        I &0 &0\\
      -B_\omega^\top &0 &0\\
        0 &I &0\\
        -I &0 &0\\
        0 &0 &I
    \end{bmatrix}\succ 0.
\end{aligned}
\end{equation}
Let $L=FH$. Using $\Phi=\begin{bmatrix}M_RL\\M_Q H\end{bmatrix}$ with $M_R^\top M_R=R, M_Q^\top M_Q=Q$, \eqref{proof7} is equivalent to
\begin{equation}\label{proof8}
\begin{aligned}
    \begin{bmatrix}
        0\\
        CH+DL\\
        \begin{bmatrix}
            H\\
            L
        \end{bmatrix}
    \end{bmatrix}
    (-H+\Phi^\top\Phi)^{-1}
    \begin{bmatrix}
        0\\
        CH+DL\\
        \begin{bmatrix}
            H\\
            L
        \end{bmatrix}
    \end{bmatrix}^\top+\\
    \begin{bmatrix}\star\end{bmatrix}^\top
    \begin{bmatrix}
        H & &\\
        &\Pi^{-1} &\\
        & &M(\tilde{\Pi})^{-1}
    \end{bmatrix}
    \begin{bmatrix}
        I &0 &0\\
      -B_\omega^\top &0 &0\\
        0 &I &0\\
        -I &0 &0\\
        0 &0 &I
    \end{bmatrix}\succ 0.
\end{aligned}
\end{equation}
Using the Schur complement twice, \eqref{proof8} is equivalent to \eqref{sdp2.2}.
Thus, \eqref{sdp1.2} is equivalent to \eqref{sdp2.2}.

Multiplying both sides of the first constraint in \eqref{sdp1.3} by $H$, it is equivalent to $H\succeq \gamma L^\top S_u L$. Using the Schur complement, it is equivalent to the first constraint in \eqref{sdp2.3}.
Multiplying both sides of the second constraint in \eqref{sdp1.3} by $H$, it is equivalent to $H\succeq \gamma H S_x H$. Using the Schur complement, it is equivalent to the second constraint in \eqref{sdp2.3}.
In summary, problem \eqref{sdp1} is equivalent to the problem \eqref{sdp2}.
$\hfill\qed$
\end{pf}

We iteratively solve the problem \eqref{sdp1} and the problem \eqref{sdp2} by fixing some variables, as detailed in Algorithm~1.
At time $t$, given the state $x_t$ and a user chosen multiplier $\hat{\Pi}\in\boldsymbol{\Pi}$, we set the variable $\Pi$ in \eqref{sdp2} to be $\hat{\Pi}$.
Then, the optimization problem \eqref{sdp2} is linear in the remaining variables and can be solved effeciently.
We solve the problem \eqref{sdp2}, obtain optimal solution $H^\star_i, L^\star_i, \tilde{\Pi}^\star_i, \gamma^\star_i$ at the $i$th iteration and calculate $F^\star_i = L^\star_i (H^\star_i)^{-1}$.
Then, we solve the problem \eqref{sdp1} by fixing $F=F^\star_i, \tilde{\Pi}=\tilde{\Pi}^\star_{i}$, obtain optimal solution $P^\star_i, \Pi^\star_i, \gamma^\star_i$.
We iterate the above procedure until we reach the maximal iteration times.
Given the state $x_t$, the solution obtained from Algorithm~1 is denoted by $\gamma_{x_t}^\star, P_{x_t}^\star, F_{x_t}^\star, \Pi_{x_t}^\star, \tilde{\Pi}_{x_t}^\star$.
The computational complexity of the SDP problems \eqref{sdp1} and \eqref{sdp2} scales with the length of the dataset \eqref{data} and the dimensions of the state and input.

\begin{algorithm}[htb]\label{algorithm1}
\caption{Iteratively find the optimal solution}
\begin{algorithmic}[1]
\Statex \textbf{Input: }$x_t$, $\hat{\Pi}$
\Statex \textbf{Output: }$\gamma_{x_t}^\star=\gamma_i^\star, P_{x_t}^\star=P_i^\star, F_{x_t}^\star=F_i^\star, \Pi_{x_t}^\star=\Pi_i^\star, \tilde{\Pi}_{x_t}^\star=\tilde{\Pi}_i^\star$
\State Set $\Pi^\star_0=\hat{\Pi}$, $i=0$;
\For{$i\leq MaxIterations$}
\State Set $i=i+1$. Fix $\Pi=\Pi^\star_{i-1}$, solve the problem \eqref{sdp2} and obtain $H^\star_i, L^\star_i, \tilde{\Pi}^\star_i, \gamma^\star_i$, calculate $F^\star_i = L^\star_i (H^\star_i)^{-1}$;
\State Fix $F=F^\star_{i}, \tilde{\Pi}=\tilde{\Pi}^\star_{i}$, solve the problem \eqref{sdp1} and obtain $P^\star_i, \Pi^\star_i, \gamma^\star_i$;
\EndFor
\end{algorithmic}
\end{algorithm}

The data-driven min-max MPC problem is solved in a receding horizon fashion, see Algorithm~2.
At time $t$, the optimal state-feedback gain $F_{x_t}^\star$ of problem \eqref{sdp1} is obtained by applying Algorithm~1.
Only the first computed input $u_t=F_{x_t}^\star x_t$ is implemented.
At the next time $t+1$, we re-iterate the above procedure.
Recomputing the optimal state-feedback gain at each time step improves closed-loop performance, as illustrated in the numerical example in Section~\ref{sec4}.

\begin{algorithm}[htb]\label{algorithm2}
\caption{Data-driven min-max MPC}
\begin{algorithmic}[1]
\State At time $t=0$, measure the state $x_0$;
\State Compute $F_{x_t}^\star$ via Algorithm~1;
\State Apply the input $u_t=F_{x_t}^\star x_t$;
\State Set $t=t+1$, measure the state $x_t$ and go back to 2;
\end{algorithmic}
\end{algorithm}

\subsection{Closed-loop guarantees}\label{sec3.3}

In the following, we show that the reformulated SDPs are recursively feasible, the resulting closed-loop system is exponentially stable and the input and state constraints are satisfied.

\begin{mythm}
Suppose Assumption~\ref{assumption1} holds, $B_\omega\neq 0$ and the optimization problem \eqref{sdp1} and \eqref{sdp2} are feasible at time $t=0$. Then, i) they are feasible at any time $t\in\mathbb{N}$;
ii) the origin is exponentially stable for the closed-loop system.
iii) the closed-loop trajectory satisfies the input and state constraints.
\end{mythm}
\begin{pf}
The proof of recursive feasibility and exponential stability is similar to \cite[Theorem 3]{xie2024minmax} so we only state the main differences here.
Since \eqref{sdp1.2} holds for the optimal solution $P_{x_t}^\star, F_{x_t}^\star, \Pi_{x_t}^\star, \tilde{\Pi}_{x_t}^\star$, multiplying it from left and right by $\begin{bmatrix}x_t^\top &\omega_t^\top &\omega_{AB, t}^\top\end{bmatrix}^\top$ and its transpose, respectively, we obtain
\begin{equation}
    \begin{aligned}
    -\|x_t\|_{P_{x_t}^\star}^2+\ell(u_t, x_t)+\|x_{t+1}\|_{P_{x_{t}^\star}}^2\\+\begin{bmatrix}
    \omega_t\\
    z_t
    \end{bmatrix}^\top \!\!\!\!\Pi_{x_t}^\star\!\!\begin{bmatrix}
    \omega_t\\
    z_t
    \end{bmatrix}
    +\begin{bmatrix}
        \omega_{AB, t}\\
        z_{AB, t}
    \end{bmatrix}^\top \!\!\!\!M(\tilde{\Pi}_{x_t}^\star)\!\!
    \begin{bmatrix}
        \omega_{AB, t}\\
        z_{AB, t}
    \end{bmatrix}\!\!\leq\! 0.
\end{aligned}
\end{equation}
Using the Assumption~\ref{assumption1} and the same arguments as in \eqref{proof4}-\eqref{thm1proof1}, we have
\begin{equation}\label{thm3proof1}
    \|x_{t+1}\|_{P_{x_{t}^\star}}^2-\|x_t\|_{P_{x_t}^\star}^2\leq -\ell(u_t, x_t).
\end{equation}
Then, we can use \eqref{thm3proof1} to show recursive feasibility and exponential stability.
The proof of constraint satisfaction follows similar arguments to those in \cite[Theorem 2]{xie2024minmax} and is therefore omitted for brevity. The key idea is to show that if constraints \eqref{sdp1.3} or \eqref{sdp2.3} are satisfied, then the corresponding input and state constraints are guaranteed to hold.
$\hfill\qed$
\end{pf}

The theoretical results require that the problem \eqref{sdp1} and \eqref{sdp2} are feasible at time $t=0$.
This requirement implies that the available data \eqref{data} must be sufficiently informative to characterize an uncertainty set $\Sigma$ that admits a common stabilizing controller.
In this framework, informativity is tied to whether the data and the structural knowledge provided by the IQC multipliers adequately restricts the set of consistent systems.
Increasing the length of the data sequence $T_f$ generally shrinks the set $\Sigma$.
Besides, the initial choice of the multiplier $\hat{\Pi}$ in Algorithm 1 is important for the feasibility of the problem~\eqref{sdp2}. 


\section{Simulation}\label{sec4}

In this section, we apply the proposed data-driven min-max MPC scheme with IQC to the following nonlinear system
\begin{equation}
\begin{aligned}
    x_{t+1}&=\begin{bmatrix}
        1.1 &0.1\\
        0.1 &0.8
    \end{bmatrix}x_t+
    \begin{bmatrix}
        1\\ 0.5
    \end{bmatrix}u_t+\begin{bmatrix}
        1\\1
    \end{bmatrix}
    \omega_t,\\
    z_t&=\begin{bmatrix}1 &0\end{bmatrix}x_t,
    \omega_t=\Delta(z_t),
\end{aligned}
\end{equation}
where $\Delta(z)=0.01\sin(z)$ satisfies $\Delta(0)=0$ as well as the sector constraint $(\Delta(z)-\alpha z)(\beta z-\Delta(z))\geq 0$ for all $z\in\mathbb{R}$ with $\alpha=-0.01, \beta=0.01$.
Then, the IQC \eqref{assumption1_equation} holds for $\Delta$ with \eqref{piuncertainty}.
The system matrices $A_s, B_s$ and the function $\Delta$ are unknown, but an input-state trajectory of length $T_f=20$ is available.
The sequence input-state data is generated by choosing the input uniformly from $[-1,1]$.
Besides, we assume that $\boldsymbol{\Pi}$ is known.
We aim to design a data-driven min-max MPC controller that stabilizes the system using the input-state trajectory.
The cost weighting matrices are chosen as $Q=0.01$ and $R=0.01$.
The input and state constraint matrices are chosen as $S_u=16$ and $S_x=0.1 I$.
The initial state is $x_0=[0.5,0.5]$.

\begin{figure}
    \centering
    \includegraphics[width=0.38\textwidth]{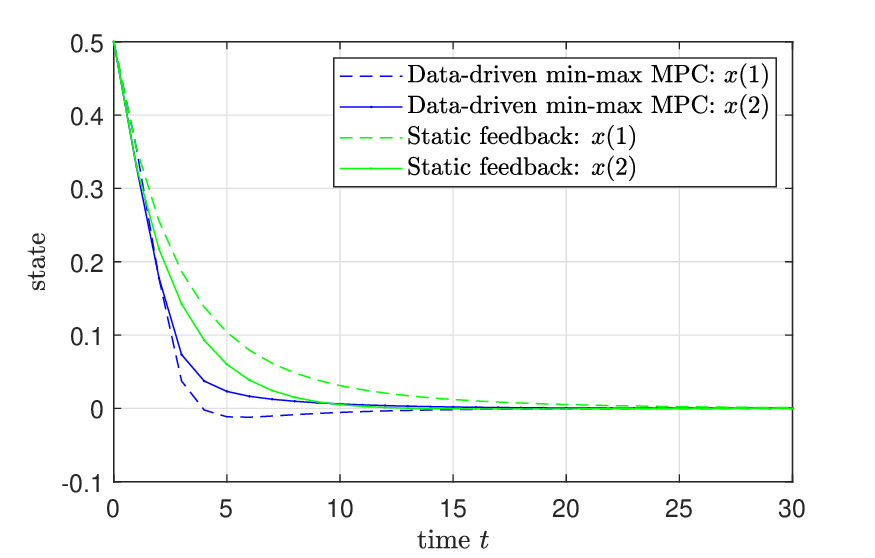}
    \includegraphics[width=0.38\textwidth]{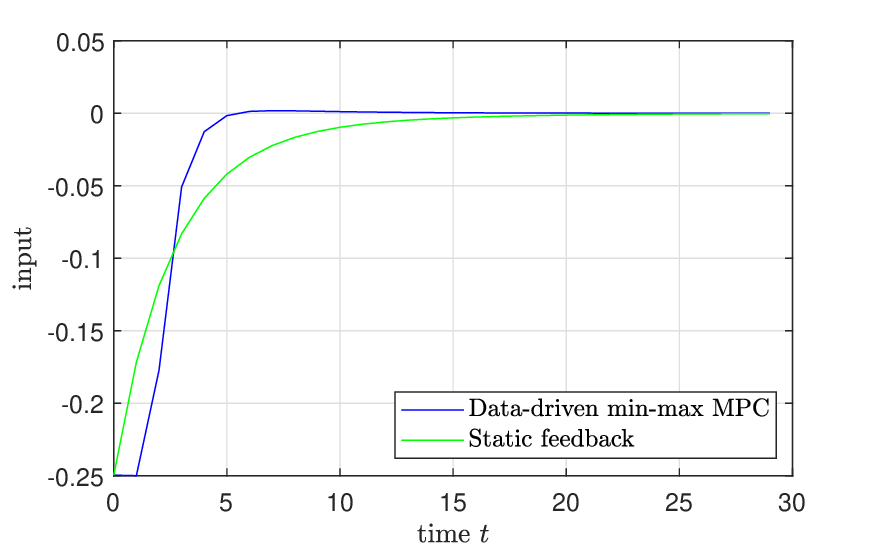}
    \caption{Closed loop state and input trajectories under the proposed data-driven min-max MPC scheme and the static state-feedback control.}
    \label{pic:compare}
\end{figure}

Figure~1 illustrates the closed-loop input and state trajectories resulting from the proposed data-driven min-max MPC scheme and a static state-feedback controller.
The static state-feedback gain is computed at the initial time $t=0$ of the data-driven min-max MPC scheme.
As shown in the figures, the input and state constraints are satisfied in closed-loop operation, i.e., the input remains above $0.25$. 
The proposed data-driven min–max MPC approach drives both the input and state trajectories to the origin more rapidly than the static state-feedback controller.
The closed-loop cost of the data-driven min-max MPC scheme is $10.8\%$ smaller than that of the static state-feedback control law, which shows that the performance is improved by recomputing the control law.

\section{Conclusion}\label{sec5}
In this paper, we proposed a data-driven min-max MPC scheme for unknown discrete-time uncertain system with IQCs.
The main contribution is to reformulate the data-driven min-max MPC to SDPs and propose an algorithm that iteratively solve two SDPs, which gives a tractable solution and yields an optimal state-feedback gain using data.
A receding-horizon algorithm is proposed to
repeatedly solve the SDPs at each time step and update the state-feedback gain online.
The resulting closed-loop system is proven to be exponentially stable and satisfy both input and state constraints.
Numerical example shows that the data-driven min-max MPC scheme is effective and improves closed-loop performance by recomputing the state-feedback gain.
In future work, we plan to extend the scheme to more general classes of IQC to handle a broader range of nonlinear dynamics.

\bibliography{ifacconf}        
\end{document}